\documentclass[preprint,12pt]{elsarticle}

\usepackage{amssymb}
\usepackage{amscd}
\usepackage{amsthm}
\usepackage{amsmath}
\usepackage{mathrsfs}
\newtheorem{theorem}{Theorem}[section]

\newtheorem{lemma}[theorem]{Lemma}
\newtheorem{corollary}[theorem]{Corollary}
\newtheorem{proposition}[theorem]{Proposition}
\newproof{Proof}{Proof}
\numberwithin{equation}{section}
\theoremstyle{definition}
\newtheorem{example}[theorem]{Example}

\begin{document}

\begin{frontmatter}



\title{Geometric Aspects of Self-adjoint Sturm-Liouville Problems}


\author{Yicao Wang}

\address{Department of Mathematics, Hohai University, Nanjing 210098, China\\}
\ead{yicwang@hhu.edu.cn}

\begin{abstract}
In the paper, we use $\mathrm{U}(2)$, the group of $2\times 2$ unitary matices, to parameterize the space of all self-adjoint boundary conditions for a fixed Sturm-Liouville equation on the interval $[0,1]$. The adjoint action of $\mathrm{U}(2)$ on itself naturally leads to a refined classification of self-adjoint boundary conditions--each adjoint orbit is a subclass of these boundary conditions. We give explicit parameterizations of those adjoint orbits of principal type, i.e. orbits diffeomorphic to the 2-sphere $S^2$, and investigate the behavior of the $n$-th eigenvalue $\lambda_n$ as a function on such orbits.
\end{abstract}


\end{frontmatter}


\section{Introduction}
Unbounded self-adjoint (SA for brevity) operators are very important objects in mathematical physics. In quantum mechanics, an observable is represented by an SA operator, rather than a symmetric one. In perturbative quantum field theory, when calculating the contribution of a one-loop graph, one should obtain the (regularized) determinant of a differential operator, but before that, a suitable SA extension should be chosen first. However, generally, there may be too many SA extensions by prescribing different SA boundary conditions. For example, consider the classical Sturm-Liouville (S-L) equation on $J=[0,1]$:
\begin{equation}ly:=-(py')'+qy=\lambda y,\quad 0<p\in C^1(J), q\in C(J). \label{SL}\end{equation}
 Then the set $\mathcal{U}$ of all complex SA boundary conditions can be divided into two mutually exclusive subsets. The first, called separated, includes boundary conditions of the form
\begin{equation}\left \{
\begin{array}{ll}
y(0)\cos \alpha-(py')(0)\sin \alpha=0,\\
 y(1)\cos \beta-(py')(1)\sin \beta=0,\end{array}
 \right.
 \label{B1}\end{equation}
 where $ \alpha \in [0, \pi), \beta\in (0,\pi]$; The second, called coupled, includes boundary conditions of the form
 \begin{equation}\left(
     \begin{array}{c}
       y(1) \\
       (py')(1) \\
     \end{array}
   \right)=e^{i\varphi}K\left(
     \begin{array}{c}
       y(0) \\
       (py')(0) \\
     \end{array}
   \right),
\label{B2} \end{equation}
 where $K\in\mathrm{SL(2, \mathbb{R})}=:\{k=\left(
            \begin{array}{cc}
              k_{11} & k_{12} \\
              k_{21} & k_{22} \\
            \end{array}
          \right); k_{ij}\in \mathbb{R}, \det K=1\}$, and $\varphi\in [0,2\pi)$.

It is well-known that the eigenvalues of the above S-L problem consisting of Eq.~(\ref{SL}) and an SA boundary condition are bounded from below and can be ordered to form a non-decreasing sequence
\[-\infty<\lambda_0\leq \lambda_1\leq \lambda_2\leq \cdots \leq \lambda_n\leq \cdots ,\]
approaching $\infty$ so that the number of times an eigenvalue can appear equals to its (geometric) multiplicity.

In this paper, we mainly put emphasis on the structure of $\mathcal{U}$ and assume that $p\equiv 1$ and $q\in C(J)$ for brevity, though the main body of results here holds for more general $p, q$'s. 

It's of interest to consider these $\lambda_n$ as functions on $\mathcal{U}$ and explore how they change when the boundary condition varies. It is already clear that, $\lambda_n$ are not continuous on $\mathcal{U}$ equipped with the natural topology \cite{Kong1}. However, when restricted on certain subset $S\subset\mathcal{U}$, $\lambda_n$ may have nice properties. For example, if on $S$ $\lambda_0$ is bounded from below, then all these $\lambda_n$ are continuous on $S$. This is called the \emph{continuity principle} in \cite{Kong1}, which we shall use frequently in the following.

An abstract theorem of von Neumann implies that $\mathcal{U}$ is \emph{globally} parameterized by $\mathrm{U}(2)$, the unitary group in complex dimension 2, but in the mathematical literature on S-L problems, in terms of boundary data, $\mathcal{U}$ is often viewed as a set of equivalence classes of matrices, say, a submanifold of the Grassmanian of 2-dimensional subspaces in $\mathbb{C}^4$. In this context, (\ref{B1}) and (\ref{B2}) are in fact preferred representative of these classes and the underlying group $\mathrm{U}(2)$ cannot be seen directly in this manner. Recently it is found in \cite{Asor} that there is, more or less, a \emph{canonical}  way to identify $\mathcal{U}$ with $\mathrm{U}(2)$ \footnote{This way of parameterizing SA extensions by $\mathrm{U}(2)$ is already known in the context of boundary triples, see for example \cite[Chap.~14]{Sch}}, which is the starting point of our paper.

  As a smooth 4-manifold, $\mathrm{U}(2)$ is very special. It is a compact Lie group and has a rich geometry. In this paper, however, we mainly consider one aspect of this geometry and its interplay with S-L problems: $\mathrm{U}(2)$ acts on itself by conjugation, i.e. $g\cdot u=gug^{-1}$ for $g, u\in \mathrm{U}(2)$. Orbits of this action are called adjoint orbits, each characterized by its eigenvalues (matrices in an orbit all have the same eigenvalues). \emph{Topologically}, these orbits are divided into two types, those consisting of a single point (the two eigenvalues are the same), and those diffeomorpic to the 2-sphere $S^2$ (the two eigenvalues are different). We shall mainly explore the behavior of $\lambda_n$ as functions on these spheres in the latter case.

 Note that in this paper, for brevity, by eigenvalues of a boundary condition $A$ (represented by a matrix) we always mean eigenvalues of the associated boundary value problem, while eigenvalues of $A$ refer to eigenvalues of the matrix $A$.

 The paper is organized as follows.

 Sec.~\ref{sec2} is divided into two subsections. In the first subsection, we discuss the structure of $\mathrm{U}(2)$ as the space of SA boundary conditions. We identify several subsets of $\mathrm{U}(2)$, parameterize them and show how these parameterizations are related to the ones given in (\ref{B1}) and (\ref{B2}). In the second subsection, we give a refined classification of SA boundary conditions in terms of adjoint orbits and parameterize orbits of \emph{principal type}--those diffeomorphic to $S^2$.

 Sec.~\ref{sec3} is devoted to briefly investigating the so-called \emph{characteristic curve} $\Gamma$ which is of great importance when one considers all SA boundary conditions together. To our knowledge, this curve was first investigated in \cite{Kong1}. The behavior of $\Gamma$ is complicated and we hardly add any new insight into it. We only rewrite it out in our context and write down \emph{the characteristic equation} in terms of it (Thm.~\ref{curve}). The advantage is that this equation is canonical and valid for all SA boundary conditions. In the end of this section, we point out the observation that $\Gamma$ has no joint point with almost all adjoint orbits of principal type. This shall imply the situation considered in Sec.~\ref{sec4} is general.

 In Sec.~\ref{sec4}, we investigate the behavior of $\lambda_n$ as a function on an adjoint orbit $\mathcal{O}$ of principal type. We show that $\lambda_n$ is continuous on $\mathcal{O}$. If, furthermore, $\Gamma$ has no joint point with $\mathcal{O}$, then $\lambda_n$ is a real analytic function on $\mathcal{O}$ and has exactly two critical points. If $[a_n, b_n]$ is the range of $\lambda_n$, then these $a_n, b_n$, $n\in \mathbb{N}$ are \emph{precisely} the zeros of a certain real analytic function and $a_n< b_n<a_{n+1}< b_{n+1}$ (Thm.~\ref{M1}, \ref{M2}, \ref{M3}). There are two viewpoints to regard eigenvalues of S-L problems: On one side, eigenvalues are roots of the characteristic equation. On the other side, eigenvalues can also be characterized in terms of quadratic forms using the \emph{min-max principle}. To obtain our results, we freely switch our viewpoint between the two if it is convenient. In the end of this section, we investigate the shape of the level subset of $\lambda$ in $\mathcal{O}$.

 The last short section can be viewed as a complement of Sec.~\ref{sec4}. We consider $\lambda_n$ as a function on the diagonal of the torus in $\mathrm{U}(2)$. We show that the range of $\lambda_n$ on $\mathrm{U}(2)$ is in fact already determined by its restriction on the diagonal (Thm.~\ref{d}).
\section{The space of SA boundary conditions and adjoint orbits}
\label{sec2}

\subsection{The space of SA boundary conditions}

Let $l_0$ and $l_1$ be the minimal and the maximal operators associated with $l$ respectively. von Neumann's abstract theory \cite[Chap.~13]{Sch} implies that the set $\mathcal{U}$ of all SA extensions of $l_0$ is parameterized by unitary transforms from $\ker (l_1-i\textup{I})$ to $\ker(l_1+i\textup{I})$. Since the two spaces are both 2-dimensional, \emph{topologically} $\mathcal{U}$ is just $\mathrm{U}(2)$. In this description, however, there is no canonical way to identify $\mathcal{U}$ with $\mathrm{U}(2)$, because to realize such a parameterization a distinguished transform should be chosen.

Recently, an explicit and canonical way of expressing SA boundary conditions in terms of elements of $\mathrm{U}(2)$ has been found \cite{Asor}. Let $y$ be a function in the Sobolev space $W_{2,2}(J)$ and \footnote{Here we use $\dot{y}$ to denote the outward unit normal derivative of $y$. So $\dot{y}(0)=-y'(0)$ and $\dot{y}(1)=y'(1)$. } \[\psi:=\left(
      \begin{array}{c}
        y(0) \\
        y(1) \\
      \end{array}
    \right), \quad
    \dot{\psi}:=\left(
      \begin{array}{c}
        \dot{y}(0) \\
        \dot{y}(1) \\
      \end{array}
    \right).
\]
An SA boundary condition then takes the following form:
\begin{equation}i(I+U)\dot{\psi}=(I-U)\psi,\label{U1}\end{equation}
where $I$ is the $2\times 2$ identity matrix. This way we shall identify $\mathcal{U}$ with $\mathrm{U}(2)$. For the details of Eq.~(\ref{U1}) and even its generalization, we refer the interested readers to \cite{Asor}.

Before proceeding further, we recall a description of $\mathrm{U}(2)$. Any element $g$ of $\mathrm{U}(2)$ can be decomposed into two factors:
\[g=\sqrt{\det g}\cdot (g/\sqrt{\det g}),\]
where $\sqrt{\det g}\in \mathrm{U}(1)$ is a square root of $\det g$, and $\frac{g}{\sqrt{\det g}}\in \mathrm{SU}(2)$, i.e. with determinant 1. Since there are two square roots of $\det g$, $\mathrm{U}(2)$ is the quotient of $\mathrm{U}(1)\times \mathrm{SU}(2)$ under the natural action of $\mathbb{Z}_2$. This result is often written as $\mathrm{U}(2)=\mathrm{U}(1)\times_{\mathbb{Z}_2} \mathrm{SU}(2)$. We denote the corresponding quotient map by $P$.

It's natural to classify all SA boundary conditions into two mutually exclusive subclasses according to whether $\det(I+U)$ equals 0 or not. Denote
\[\mathcal{U}_0=\{U\in \mathcal{U}|\det (I+U)=0\},\quad \mathcal{U}_1=\{U\in \mathcal{U}| \det(I+U)\neq0\}.\]

$\mathcal{U}_1$ is certainly open and dense in $\mathrm{U}(2)$. If $U\in \mathcal{U}_1$, then
\begin{equation}A:=-i(I+U)^{-1}(I-U)\label{Cay}\end{equation}
is actually a Hermitian matrix and precisely the Cayley transform of $U$. In terms of $A$, the boundary condition (\ref{U1}) can then be rewritten as
\begin{equation}\dot{\psi}=A\psi.\label{Cay1}\end{equation}
Note that since $A$ and $U$ are in 1-1 correspondence in $\mathcal{U}_1$, $A$ can also be viewed as the coordinate in the chart $\mathcal{U}_1\subset \mathcal{U}$ (so topologically $\mathcal{U}_1\simeq \mathbb{R}^4$). Let $y_1, y_2$ be the solutions of Eq.~(\ref{SL}), satisfying
\[y_1(0)=1,\quad y_1'(0)=0,\quad y_2(0)=0,\quad y_2'(0)=1.\]If $A=\left(\begin{array}{cc} a & b \\\bar{b} & c \\
\end{array} \right)$, where $a, c$ are real numbers and $b$ is complex, then the characteristic equation is
                                  \[\Delta(\lambda)=\det [\left(
                                    \begin{array}{cc}
                                      0 & -1 \\
                                      \dot{y}_1 & \dot{y}_2 \\
                                    \end{array}
                                  \right)-\left(
                                    \begin{array}{cc}
                                      a & b \\
                                      \bar{b} & c \\
                                    \end{array}
                                  \right)\left(
                                    \begin{array}{cc}
                                      1 & 0 \\
                                      y_1 & y_2 \\
                                    \end{array}
                                  \right)]=0,\]
i.e.
\begin{equation}-a\dot{y}_2+(ac-|b|^2)y_2-cy_1+\dot{y}_1-2\Re b=0.\label{chh}\end{equation}

Let's come to the structure of $\mathcal{U}_0$. If $U\in \mathcal{U}_0$, we can set
$U=e^{i\theta}\left(
                                    \begin{array}{cc}
                                      a & b \\
                                      -\bar{b} & \bar{a} \\
                                    \end{array}
                                  \right)$,
where $\theta\in[0, \pi]$, and $\left(
                                    \begin{array}{cc}
                                      a & b \\
                                      -\bar{b} & \bar{a} \\
                                    \end{array}
                                  \right)\in \mathrm{SU}(2)$. Let $a=re^{i\beta}, r\in[0, 1], \beta\in[0, 2\pi)$. Then one can find that
\begin{equation}e^{i\theta}=-r\cos \beta+i\sqrt{1-r^2\cos^2 \beta }.\label{th}\end{equation}
So $U$ is completely determined by its factor in $\mathrm{SU}(2)$. But $\pm I\in \mathrm{SU}(2)$ determine the same $U=-I$. This argument shows that $\mathcal{U}_0$ is topologically the 3-sphere $S^3$ with two points glued together.\footnote{ Topologically $\mathrm{SU}(2)\simeq S^3$.} A general element of $\mathcal{U}_0$ is of the following form:
\[e^{i\theta}\left(\begin{array}{cc}
                                      re^{i\beta} & \sqrt{1-r^2}e^{i\gamma} \\
                                      -\sqrt{1-r^2}e^{-i\gamma}& re^{-i\beta} \\
                                    \end{array}
                                  \right),\]
where $\theta$ is given by Eq.~(\ref{th}) and $r\in [0,1], \beta, \gamma\in [0, 2\pi)$.

There is another interesting subset $\mathcal{U}^{\mathbb{R}}\subset\mathcal{U}$, consisting of all real SA boundary conditions. As for the shape of $\mathcal{U}^{\mathbb{R}}$ in $\mathrm{U}(2)$, we have
\begin{proposition}$\mathcal{U}^{\mathbb{R}}=\{U\in \mathrm{U}(2)|U=U^t\}$, where the superscript $t$ denotes the transpose of a matrix. Furthermore, with $\mathrm{U}(2)$ viewed as $\mathrm{U}(1)\times_{\mathbb{Z}_2} \mathrm{SU}(2)$, $\mathcal{U}^{\mathbb{R}}$ is topologically just $P(S^1\times S^2)$ (for the precise meaning, see the proof).
\end{proposition}
\begin{proof}
A real SA boundary condition is precisely one whose complex conjugate represents the same boundary condition except that $\psi, \dot{\psi}$ are replaced by $\bar{\psi}, \bar{\dot{\psi}}$. The complex conjugate of (\ref{U1}) is
\begin{equation*}i(I+\bar{U})\bar{\dot{\psi}}=-(I-\bar{U})\bar{\psi}.\end{equation*}
It can be rewritten as
\begin{equation*}i\bar{U}(I+\bar{U}^{-1})\bar{\dot{\psi}}=\bar{U}(I-\bar{U}^{-1})\bar{\psi},\end{equation*}
i.e.
\begin{equation*}i(I+\bar{U}^{-1})\bar{\dot{\psi}}=(I-\bar{U}^{-1})\bar{\psi}.\end{equation*}
Then reality means $U=\bar{U}^{-1}$, which is precisely $U=U^t$.
Let $U=e^{i\theta}\left(
                                    \begin{array}{cc}
                                      a & b \\
                                      -\bar{b} & \bar{a} \\
                                    \end{array}
                                  \right)$. Then reality precisely means $b$ is purely imaginary or zero. This observation immediately leads to the conclusion that $\mathcal{U}^{\mathbb{R}}=P(S^1\times S^2)$.

\end{proof}
\noindent \emph{Remark}. In \cite[Thm~3.3]{Kong2}, there is also a description of the space of all real SA boundary conditions. However, the global picture is more clear here.

Let's now see how the boundary conditions (\ref{B1}) (\ref{B2}) look like in $\mathrm{U}(2)$. It is easy to see that the separated boundary conditions correspond to $U$'s of diagonal form, forming a Cartan subgroup $\mathrm{H}$ of $\mathrm{U}(2)$, topologically a 2-torus.

 For the coupled case, two subcases should be distinguished, $k_{12}\neq 0$ and $k_{12}=0$.

\begin{proposition}\label{CH}If in the coupled case $k_{12}\neq 0$, then the corresponding $U$ lies in $\mathcal{U}_1$ and the associated Hermitian matrix is
\[A(e^{i\varphi}K)=\frac{1}{k_{12}}\left(\begin{array}{cc}
   k_{11} & -e^{-i\varphi} \\
  -e^{i\varphi}& k_{22} \\
\end{array}\right).\]
\end{proposition}
\begin{proof}If $k_{12}\neq 0$, the boundary condition Eq.~(\ref{B2}) can be rewritten as
\[\dot{\psi}=\left(\begin{array}{cc}
   k_{11}/k_{12} & -e^{-i\varphi}/k_{12} \\
  -e^{i\varphi}/k_{12} & k_{22}/k_{12} \\
\end{array}\right)\psi.\]
Comparing this with Eq.~(\ref{Cay1}), we come to the conclusion.
\end{proof}
It is not hard to see that if $k_{12}=0$, the corresponding boundary condition cannot be rewritten as Eq.~(\ref{Cay1}) and so the corresponding $U(e^{i\varphi}K)\in \mathcal{U}_0$. However, from the above proposition, we can obtain a unified expression of $U(e^{i\varphi}K)$ no matter whether $k_{12}=0$ or not:
\begin{proposition}
For the coupled boundary condition (\ref{B2}), the corresponding element $U(e^{i\varphi}K)\in \mathrm{U}(2)$ is
\[\frac{1}{k_{12}-k_{21}+i(k_{11}+k_{22})}\left(\begin{array}{cc}
  k_{12}+k_{21}+i(k_{22}-k_{11}) & 2ie^{-i\varphi} \\
  2ie^{i\varphi} & k_{12}+k_{21}-i(k_{22}-k_{11}) \\
\end{array}\right),\]
and it has $-1$ as its eigenvalue if and only if $k_{12}=0$.
\end{proposition}
\begin{proof}
If $k_{12}\neq 0$, then from Prop.~\ref{CH}, $U(e^{i\varphi}K)\in \mathcal{U}_1$ and
\[U(e^{i\varphi}K)=[I-iA(e^{i\varphi}K)][I+iA(e^{i\varphi}K)]^{-1}\]
due to Eq.~(\ref{Cay}). This leads to the expression as required. Obviously this expression extends smoothly to the case $k_{12}=0$.

That $\det [I+U(e^{i\varphi}K)]=0$ is equivalent to
\[k_{12}[k_{12}-k_{21}+i(k_{22}+k_{11})]=0,\]
which holds if and only if $k_{12}=0$.
\end{proof}
For the case $k_{12}=0$, we also have
\begin{proposition}
If in the coupled case $k_{12}=0$, then in terms of $r, \beta, \gamma$, the matrix $e^{i\varphi}K$ is determined by (without loss of generality, we set $k_{11}>0$)
\[k_{11}=\frac{\sqrt{1-r^2\cos^2 \beta}+r\sin \beta}{\sqrt{1-r^2}},\]
\[k_{21}=\frac{-2r\cos \beta}{\sqrt{1-r^2}},\]
\[e^{i\varphi}=e^{-i(\gamma+\frac{\pi}{2})}.\]
\end{proposition}
\begin{proof}
In terms of $r, \beta, \gamma$, the associated boundary condition is
\begin{equation*}\left \{
\begin{array}{ll}
y(1)=\frac{\sqrt{1-r^2\cos^2 \beta}+r \sin \beta}{\sqrt{1-r^2}}e^{-i(\gamma+\frac{\pi}{2})}y(0),\\
 y'(1)=-\frac{2r\cos \beta}{\sqrt{1-r^2}}e^{-i(\gamma+\frac{\pi}{2})}y(0)+\frac{\sqrt{1-r^2\cos^2 \beta}-r \sin \beta}{\sqrt{1-r^2}}e^{-i(\gamma+\frac{\pi}{2})}y'(0).\end{array}
 \right.
 \end{equation*}
 Comparing this with Eq.~(\ref{B2}), we get the conclusion.
\end{proof}
$\mathcal{U}_0$ has been investigated in other way in the literature. From the above discussion, it is easy to see that $\mathcal{U}_0$ is actually the set $\mathscr{J}^\mathbb{C}$  in \cite{Kong1}.\\
\noindent \emph{Remark}. No matter whether $k_{12}=0$ or not, the eigenvalues of $U(e^{i\varphi}K)$ are independent of $\varphi$. So for a fixed $K$, $U(e^{i\varphi}K),\varphi\in [0,2\pi)$ all lie in the same adjoint orbit, tracing out a circle. There is a beautiful inequality among eigenvalues of S-L problems when the boundary condition varies only on this circle \cite{Eas}.
\subsection{Adjoint orbits}
Let $\mathrm{H}\subset \mathrm{U}(2)$ be the Cartan subgroup as in the last subsection, and $W$ ($\cong\mathbb{Z}_2$) the corresponding Weyl group. Then the quotient $\mathrm{H}/W$ is a 2-dimensional manifold with boundary--in fact, it is topologically the famous Mobius strip. $\mathrm{H}/W$ can be viewed as the space of adjoint orbits in $\mathrm{U}(2)$, with each interior point representing an adjoint orbit of principal type and with each point on the boundary representing an adjoint orbit consisting of a single matrix. In this sense, a generic adjoint orbit is diffeomorphic to $S^2$. Let $\Pi: \mathrm{U}(2)\rightarrow \mathrm{H}/W$ be the quotient map. We refer the reader to \cite{Sep} for the basics of compact Lie group.

Since both $\mathcal{U}_0$ and $\mathcal{U}_1$ are invariant under the adjoint action, an adjoint orbit would lie entirely either in $\mathcal{U}_0$ or $\mathcal{U}_1$. This, of course, leads to a more refined classification of SA boundary conditions--each adjoint orbit represents a subclass. In this subsection, we mainly consider orbits of principal type. These are in fact real analytic 2-dimensional manifolds.

By (\ref{Cay}), $\mathcal{U}_1$ is diffeomorphic to the space $\mathcal{M}$ of $2\times 2$ Hermitian matrices and, the adjoint action of $\mathrm{U}(2)$ on $\mathcal{U}_1$ corresponds to the one on $\mathcal{M}$. This way, we can identify adjoint orbits in $\mathcal{U}_1$ with adjoint orbits in $\mathcal{M}$. An adjoint orbit $\mathcal{O}\subset \mathcal{M}$ is characterized by its eigenvalues $\zeta_1>\zeta_2$. Let $\mu=\frac{\zeta_1+\zeta_2}{2}$, $\nu=\frac{\zeta_1-\zeta_2}{2}$ and denote the adjoint orbit by $\mathcal{O}_{\mu,\nu}$.

\begin{proposition}A general element in $\mathcal{O}_{\mu,\nu}$ is of the following form:
\[A=\left( \begin{array}{cc}
\mu-\nu\cos 2\theta & \nu \sin 2\theta \cdot e^{-i\gamma}\\
\nu \sin 2\theta\cdot e^{i\gamma} & \mu+\nu\cos 2\theta \\
 \end{array}
    \right),\quad\quad \gamma\in [0, 2\pi), \quad \theta\in [0, \frac{\pi}{2}].\]
\end{proposition}
\begin{proof}
$\mathcal{O}_{\mu,\nu}$ is the adjoint orbit through $\left( \begin{array}{cc}
\mu-\nu & 0\\
0 & \mu+\nu \\
 \end{array}
    \right)$. Then each element in $\mathcal{O}_{\mu,\nu}$ can be represented by
   \[\left(
                                    \begin{array}{cc}
                                      a & b \\
                                      -\bar{b} & \bar{a} \\
                                    \end{array}
                                  \right)\left( \begin{array}{cc}
\mu-\nu & 0\\
0 & \mu+\nu \\
 \end{array}
    \right)\left(
                                    \begin{array}{cc}
                                      a & b \\
                                      -\bar{b} & \bar{a} \\
                                    \end{array}
                                  \right)^{-1},\]
                                  for some
    $\left(\begin{array}{cc}
                                      a & b \\
                                      -\bar{b} & \bar{a} \\
                                    \end{array}
                                  \right)\in \mathrm{SU}(2)$. We can even further require $a\geq 0$. Setting $a=\cos \theta$, $\theta\in [0, \frac{\pi}{2}]$,
 and $b=\sin \theta \cdot e^{-i\gamma}$ then leads to the representation.
\end{proof}
\noindent \emph{Remark}. From Prop.~\ref{CH}, we can see that $\gamma$ essentially contains the same geometric content as $\varphi$ in Eq.~(\ref{B2}). Note that $\theta=0,\frac{\pi}{2}$ actually correspond to the only two separated boundary conditions in $\mathcal{O}_{\mu,\nu}$. It is easy to find that in $\mathcal{O}_{\mu,\nu}$, real boundary conditions lie precisely on the circle formed by the two semi-circles $\gamma=0$ and $\gamma=\pi$. It will soon be clear that this is a general property of orbits of principal type.

An adjoint orbit in $\mathcal{U}_0$ is determined by the other eigenvalue $e^{i\alpha}$ ($\alpha\in[0,\pi)\cup (\pi, 2\pi)$) besides $-1$. We shall denote the orbit by $\mathcal{O}_{\alpha}$. Since in this case,
\[\textup{tr}U=e^{i\theta}(a+\bar{a})=-1+e^{i\alpha},\]
we find
$\Re a=\sin \frac{\alpha}{2}$ if $\alpha\in [0, \pi)$ and, $\Re a=-\sin \frac{\alpha}{2}$ if $\alpha\in (\pi, 2\pi)$.

\begin{proposition}For $\alpha\in [0, \pi)$, a general element of $\mathcal{O}_\alpha$ is of the following form:
\[U=ie^{i\alpha/2}\left(\begin{array}{cc}
 \sin\frac{\alpha}{2}+i t & \sqrt{\cos^2\frac{\alpha}{2}-t^2}e^{i\gamma} \\
  -\sqrt{\cos^2\frac{\alpha}{2}-t^2}e^{-i\gamma}& \sin\frac{\alpha}{2}-i t \\
 \end{array}
 \right),\]
 where $t\in [-\cos \frac{\alpha}{2}, \cos \frac{\alpha}{2}]$ and $\gamma\in [0, 2\pi)$.

 For $\alpha\in (\pi, 2\pi)$, a general element of $\mathcal{O}_\alpha$ is of the following form:
 \[U=-ie^{i\alpha/2}\left(\begin{array}{cc}
 -\sin\frac{\alpha}{2}+i t & \sqrt{\cos^2\frac{\alpha}{2}-t^2}e^{i\gamma} \\
  -\sqrt{\cos^2\frac{\alpha}{2}-t^2}e^{-i\gamma}& -\sin\frac{\alpha}{2}-i t \\
 \end{array}
 \right),\]
 where $t\in [\cos \frac{\alpha}{2}, -\cos \frac{\alpha}{2}]$ and $\gamma\in [0, 2\pi)$.
 \end{proposition}
 \begin{proof} By Eq.~(\ref{th}),
 \[e^{i\theta}= -\sin\frac{\alpha}{2}+i\cos \frac{\alpha}{2}=i e^{i\alpha/2}\]
 if $\alpha\in[0, \pi)$, and
 \[e^{i\theta}= \sin\frac{\alpha}{2}-i\cos \frac{\alpha}{2}=-i e^{i\alpha/2}\]
 if $\alpha\in (\pi, 2\pi)$. The conclusion then easily follows.
 \end{proof}

 \section{The characteristic curve}
 \label{sec3}
The characteristic curve $\Gamma: \mathbb{R}\rightarrow \mathrm{U}(2)$ is a parameterized curve, the image of which consists of all SA boundary conditions having a double eigenvalue. This curve contains all information concerning eigenvalues of SA boundary conditions (of course, if one puts eigenfunctions aside) \cite{Kong2}.

From (\ref{U1}), it is easy to find that $\Gamma$ is of the following form :
\[\Gamma(\lambda)=\frac{1}{y_2-\dot{y}_1+i\dot{y}_2+iy_1}\left(
                                    \begin{array}{cc}
                                      y_2+\dot{y}_1+i\dot{y}_2-iy_1 & 2i \\
                                      2i & y_2+\dot{y}_1-i\dot{y}_2+iy_1\\
                                    \end{array}
                                  \right),\]
where $\lambda\in \mathbb{R}$. The image of $\Gamma$ is completely included in $\mathcal{U}^{\mathbb{R}}$. $\Pi\circ \Gamma$ is a curve in $\mathrm{H}/W$. We call $\Pi\circ \Gamma$ the \emph{induced curve} of $\Gamma$. $\Pi\circ \Gamma$ is characterized by the two eigenvalues of $\Gamma(\lambda)$, say,
\[\kappa_\pm(\lambda)=\frac{y_2+\dot{y}_1\pm i\sqrt{4+(\dot{y}_2-y_1)^2}}{y_2-\dot{y}_1+i\dot{y}_2+i y_1}.\]

\begin{proposition}\label{curve}In terms of $\Gamma(\lambda)$, the characteristic equation for an SA boundary condition $U$ can be written in the following form:
\begin{equation}\det(U-\Gamma(\lambda))=0. \label{ce}\end{equation}
The subset $S_\lambda\subset \mathcal{U}$ of boundary conditions with $\lambda$ as an eigenvalue is diffeomorphic to $\mathcal{U}_0$.
\end{proposition}
\begin{proof}
$\mathrm{U}(2)$ acts on itself by left translation and $S_\lambda$ can be represented as $-\Gamma(\lambda)\mathcal{U}_0$. By Eq.~(\ref{ce}), $S_\lambda$ is diffeomorphic to $\mathcal{U}_0$, i.e. a 3-sphere with 2 points glued together. This observation was already noted in \cite{Kong2}, but in a more complicated language.
\end{proof}
\begin{corollary}The matrix $\Gamma(\lambda)$ has $-1$ as an eigenvalue if and only if $\lambda=\lambda_n^D$ for some $n\in \mathbb{N}$. Therefore, the characteristic curve $\Gamma$ intersects $\mathcal{U}_0$ countably infinite times.
\end{corollary}
\begin{proof}
This is obvious.
\end{proof}
\noindent\emph{Remark}. However, the above result doesn't mean that $\Gamma$ has infinitely many intersection points with $\mathcal{U}_0$. Besides, $-1$ can be replaced by $e^{i\theta}I$, $\theta\in[0, 2\pi)$, and similar result holds.
\begin{example}Let $q\equiv 0$. Then for $\lambda> 0$, $y_1(x,\lambda)=\cos\sqrt{\lambda}x, y_2(x,\lambda)=\frac{\sin\sqrt{\lambda}x}{\sqrt{\lambda}}$. The two eigenvalues of $\Gamma(\lambda)$ are
\[\kappa_\pm(\lambda)=\frac{(\frac{1}{\sqrt{\lambda}}-\sqrt{\lambda})\sin \sqrt{\lambda}\pm 2 i}{(\frac{1}{\sqrt{\lambda}}+\sqrt{\lambda})\sin \sqrt{\lambda}+ 2i\cos \sqrt{\lambda}}.\]
For the Dirichlet boundary condition, $\lambda_n^D=(n+1)^2\pi^2$, $\kappa_\pm(\lambda_n^D)=\pm (-1)^{n+1}$. So in this case, the intersection points of $\Gamma$ and $\mathcal{U}_0$ all lie in the orbit through $\left( \begin{array}{cc}
1 & 0\\
0 & -1 \\
 \end{array}
    \right)$. In fact, there are only two such points, i.e. $\pm\left( \begin{array}{cc}
0 & 1\\
1 & 0 \\
 \end{array}
    \right)$.
\end{example}

Since the space $\mathrm{H}/W$ parameterizing all adjoint orbits is of dimension 2, and the induced curve of $\Gamma$ is analytic and, of course, of dimension 1, the characteristic curve $\Gamma$ would not go through a generic adjoint orbit of principal type.

\section{$\lambda_n$ as functions on adjoint orbits of principal type}
\label{sec4}
In this section, by adjoint orbits we will always refer to those of principal type.
We mainly consider adjoint orbits which have no joint point with the characteristic curve $\Gamma$. From the last section, we know a generic adjoint orbit is of this kind. By $\lambda_n^N$, we denote the $n$-th eigenvalue of the Neumann boundary condition.

For the orbit $\mathcal{O}_{\mu,\nu}$, by Eq.~(\ref{chh}) the corresponding characteristic equation is
\begin{equation}(\mu-\nu\cos 2\theta)\dot{y}_2+(\nu^2-\mu^2)y_2+(\mu+\nu\cos 2\theta)y_1-\dot{y}_1=-2\nu \sin 2\theta\cdot \cos \gamma.\label{c1}\end{equation}

\begin{lemma}\label{lm1}
Let $\lambda_n^{\pm}$ be the $n$-th eigenvalues of the boundary conditions
\[\dot{y}(0)=(\mu-\nu)y(0), \quad \dot{y}(1)=(\mu-\nu)y(1)\]
and
\[\dot{y}(0)=(\mu+\nu)y(0), \quad \dot{y}(1)=(\mu+\nu)y(1)\]
respectively. Then the function $\lambda_n$ on $\mathcal{O}_{\mu,\nu}$ satisfies
\[\lambda_n^-\leq \lambda_n\leq \lambda_n^+.\]
In particular, by the continuity principle, $\lambda_n$ is continuous on $\mathcal{O}_{\mu,\nu}$.
\end{lemma}
\begin{proof}
For $A\in \mathcal{O}_{\mu,\nu}$, the associated quadratic form is
\[Q(y)=\int_0^1|y'|^2 dx+\int_0^1q(x)|y|^2dx-\psi^\dag A \psi,\quad y\in H^1,\]
where $H^1$ is the Sobolev space $W_{1,2}(J)$.

Note that
\[\psi^\dag A \psi=\mu|\psi|^2+(|y(1)|^2-|y(0)|^2)\nu\cos 2\theta+2\nu\Re[\bar{y}(0)y(1)e^{-i\gamma}]\sin 2\theta.\]
By the inequality $2|ab|\leq |a|^2+|b|^2$, we have
\[2\Re[\bar{y}(0)y(1)e^{-i\gamma}]\sin 2\theta\leq (1+\cos 2\theta)|y(0)|^2+(1-\cos 2\theta)|y(1)|^2\]
and
\[2\Re[\bar{y}(0)y(1)e^{-i\gamma}]\sin 2\theta\geq -(1-\cos 2\theta)|y(0)|^2-(1+\cos 2\theta)|y(1)|^2.\]
Therefore, we come to the estimation
\[(\mu-\nu)|\psi|^2\leq\psi^\dag A \psi\leq (\mu+\nu)|\psi|^2.\]
The conclusion then follows from the variational characterization of $\lambda_n(A)$--the min-max principle.
\end{proof}
 \noindent \emph{Remark}. The boundedness from below of $\lambda_n$ on $\mathcal{O}_{\mu,\nu}$ is actually a conclusion of \cite{Kong1} that $\lambda_n$ is continuous on $\mathcal{U}_1$, together with the fact that $S^2$ is compact. Conversely, minor modification of the proof of Lemma \ref{lm1} gives another proof of that $\lambda_n$ is continuous on $\mathcal{U}_1$.
 \begin{proposition}
 $\lambda_n$ is a continuous function on $\mathcal{U}_1$.
 \end{proposition}
 \begin{proof}
 For any given $A_0\in \mathcal{U}_1$, let $\mathcal{O}_{\mu_0,\nu_0}$ be the orbit through $A_0$ (we allow $\nu_0$ to be 0 here). Then for $\delta>0$, the set
 \[V_\delta=\cup_{\mu+\nu<\mu_0+\nu_0+\delta} \mathcal{O}_{\mu,\nu}\]
 is an open neighbourhood of $A_0$. Note that for $A\in V_\delta$,
 \[(\mu_0+\nu_0+\delta)|\psi|^2\geq\psi^\dag A \psi.\]
 The min-max principle implies that $\lambda_0$ is bounded from below on $V_\delta$ and thus $\lambda_n$ is continuous on $V_\delta$, and in particular, continuous at $A_0$.
 \end{proof}
\begin{theorem}\label{M1} Assume that $\mathcal{O}_{\mu,\nu}$ has no joint point with $\Gamma$. Then for each $n$, $\lambda_n$ as a function on $\mathcal{O}_{\mu,\nu}$ is real analytic, and has exactly two critical points. Let $[a_n, b_n]$ be the range of $\lambda_n$ on $\mathcal{O}_{\mu,\nu}$. Then for each $n$,
\[a_n< b_n<a_{n+1}<b_{n+1}.\]
These $a_n, b_n$, $n=0,1,2,\cdots$ are exactly roots of the following equation:
\begin{equation}\nu^2(\dot{y}_2-y_1)^2+4\nu^2=[\mu(\dot{y}_2+y_1)+(\nu^2-\mu^2)y_2-\dot{y}_1]^2.\label{cr1}\end{equation}
\end{theorem}
\begin{proof}
Denote the LHS of Eq.~(\ref{c1}) by $D(\lambda, p)$, viewed as a function on $\mathbb{R}\times \mathcal{O}_{\mu,\nu}$. Since $\mathcal{O}_{\mu,\nu}$ has no joint point with $\Gamma$, $\frac{\partial D}{\partial \lambda}|_{\lambda_n(p), p}\neq 0$ for any $p\in \mathcal{O}_{\mu,\nu}$. Besides, $D(\lambda, p)$ and the RHS of Eq.~(\ref{c1}) are real analytic functions on $\mathbb{R}\times \mathcal{O}_{\mu,\nu}$. So by the implicit function theorem, $\lambda_n$ is a real analytic function on $\mathcal{O}_{\mu,\nu}$.

It is easy to find that for a critical point $p$, we must have $\sin \gamma=0$. This implies that all critical points must lie on the circle $C_0$ formed by the two semi-circles $\gamma=0$ and $\gamma=\pi$. So, to find all critical points of $\lambda_n$ on $\mathcal{O}_{\mu,\nu}$, we only need to find all critical points of $\lambda_n$ on $C_0$. Now consider the characteristic equation restricted on $C_0$, i.e.
\begin{equation}\nu(\dot{y}_2-y_1)\cos 2\theta-2\nu\sin 2\theta=\mu(\dot{y}_2+y_1)+(\nu^2-\mu^2)y_2-\dot{y}_1, \label{cr}\end{equation}
where $\theta\in (-\frac{\pi}{2}, \frac{\pi}{2}]$. For a given $\lambda\in \mathbb{R}$, there are at most two values of $\theta$ satisfying the above equation. It is an element calculation to show that $\lambda_n$ has no degenerate critical point. These together imply that there are at most two critical points of $\lambda_n$ as a function on $C_0$. Since $C_0$ is compact, we know that there are \emph{precisely} two critical points, one the maximizer and the other the minimizer.

Any critical value $\kappa$ of $\lambda_n$ must satisfy Eq.~(\ref{cr1}). Conversely, it is not hard to find that any root $\kappa$ of Eq.~(\ref{cr1}) must be a critical value of some $\lambda_n$. By the uniqueness of minimizer and maximizer, $\kappa=a_n$ or $b_n$.

If $a_{n+1}=\lambda_{n+1}(p_0)$ for some $p_0\in C_0$, then
\[a_{n+1}>\lambda_n(p_0)\geq a_n.\]

We only need to check that $a_{n+1}> b_n$. If it is not the case, then $a_{n+1}\in (a_n, b_n]$ and there is another point $p_1\in C_0$ such that \[\lambda_n(p_1)=\lambda_{n+1}(p_0)=a_{n+1}.\]
If $p_1=p_0$, this means that $a_{n+1}$ is a double eigenvalue of the boundary condition $p_0$, contradicting that $\Gamma$ has no joint point with $\mathcal{O}_{\mu, \nu}$; if $p_1\neq p_0$, then for $\lambda=a_{n+1}$, Eq.~(\ref{cr}) of $\theta$ has at least two different solutions. This contradicts the fact that $a_{n+1}$ is the unique minimum of $\lambda_{n+1}$. The proof is then completed.

\end{proof}
\begin{example}
Let $q\equiv 0$. Then for $\lambda>0$ Eq.~(\ref{cr}) is
\[-2\nu\sin 2\theta=2\mu\cos \sqrt{\lambda}+(\nu^2-\mu^2)\frac{\sin\sqrt{\lambda} }{\sqrt{\lambda}}+\sqrt{\lambda}\cos \sqrt{\lambda}.\]
$\theta=\pm\frac{\pi}{4}$ are the common critical points of all $\lambda_n$ such that $\lambda_n>0$. Eq.~(\ref{cr1}) now is
\[2\mu\cos \sqrt{\lambda}+(\nu^2-\mu^2)\frac{\sin\sqrt{\lambda} }{\sqrt{\lambda}}+\sqrt{\lambda}\cos \sqrt{\lambda}=\pm 2\nu.\]
If $\mu=\nu$, then the above equation obtains a more accessible form:
\[\cos \sqrt{\lambda}=\pm\frac{2\nu}{2\nu+\sqrt{\lambda}}.\]
\end{example}

From \cite[Thm.~3.73]{Kong1}, we can derive that, nearly all points in $\mathcal{U}_0$ are discontinuity points of $\lambda_n$ as a function on $\mathrm{U}(2)$. This, of course, doesn't exclude the possibility that $\lambda_n$ is continuous on adjoint orbits lying in $\mathcal{U}_0$.

For the orbit $\mathcal{O}_\alpha$ with $\alpha\in[0,\pi)$, the associated characteristic equation is
 \[(-\cos \frac{\alpha}{2}+t)\dot{y}_2+(-\cos\frac{\alpha}{2}-t)y_1-2\sin \frac{\alpha}{2} y_2=2\cos (\gamma+\frac{\pi}{2})\sqrt{\cos^2\frac{\alpha}{2}-t^2}.\]

\begin{lemma}\label{i}On the orbit $\mathcal{O}_\alpha$ with $\alpha\in[0, \pi)$,
\[\lambda_n\geq\lambda_n^N,\quad n=0,1,2,\cdots.\]
In particular, by the continuity principle, $\lambda_n$ is continuous on $\mathcal{O}_\alpha$.
\end{lemma}
\begin{proof}If $t\neq \cos \frac{\alpha}{2}$, the associated quadratic form is
\[Q_1(y)=\int_0^1|y'|^2 dx+\int_0^1q(x)|y|^2dx+\frac{2\sin \frac{\alpha}{2}}{\cos \frac{\alpha}{2}-t}|y(0)|^2, y\in H_{\gamma,t}^1, \]
where \[H_{\gamma,t}^1=\{y\in H^1|y(1)=e^{-i(\gamma+\frac{\pi}{2})}\sqrt{\frac{\cos \frac{\alpha}{2}+t}{\cos \frac{\alpha}{2}-t}}y(0)\}\subset H^1.\]

Note that in this case, by the min-max principle,
\[\lambda_n=\min_{S_{n+1}\subset H^1_{\gamma, t},}\max_{y\in S_{n+1}-\{0\}}\frac{Q_1(y)}{\|y\|^2},\]
where $S_{n+1}$ ranges over all $n+1$-dimensional subspaces of $H^1_{\gamma, t}$.
Since
\[\min_{S_{n+1}\subset H^1_{\gamma, t},}\max_{y\in S_{n+1}-\{0\}}\frac{Q_1(y)}{\|y\|^2}\geq \min_{S_{n+1}\subset H^1_{\gamma, t}}\max_{y\in S_{n+1}-\{0\}}\frac{Q_0(y)}{\|y\|^2}\]
where $Q_0(y)=\int_0^1|y'|^2 dx+\int_0^1q(x)|y|^2dx$, and
\[\lambda_n^N=\min_{S_{n+1}\subset H^1}\max_{y\in S_{n+1}-\{0\}}\frac{Q_0(y)}{\|y\|^2}\]
where $S_{n+1}$ ranges over all $n+1$-dimensional subspaces of $H^1$, we must have
\[\lambda_n\geq\lambda_n^N.\]

If $t=\cos \frac{\alpha}{2}$, the associated quadratic form is
\[Q_2(y)=\int_0^1|y'|^2 dx+\int_0^1q(x)|y|^2dx+\tan \frac{\alpha}{2}|y(1)|^2,\quad y\in H^1, y(0)=0.\]
A similar argument then leads to the inequality $\lambda_n\geq \lambda_n^N.$
\end{proof}
\begin{corollary}Let $\Omega:=\cup_{\alpha\in [0,\pi)}\mathcal{O}_\alpha$. Then $\lambda_n$ are continuous functions on $\Omega$.
\begin{proof}
Note that $\lambda_0^N$ is independent of the orbit parameter $\alpha$. The inequality in Lemma \ref{i} holds uniformly on $\Omega$. By the continuity principle, the conclusion follows.
\end{proof}

\end{corollary}
\begin{theorem}\label{M2}Assume that $\mathcal{O}_\alpha$ with $\alpha\in[0, \pi)$ has no joint point with $\Gamma$. Then on $\mathcal{O}_\alpha$, $\lambda_n$ is real analytic, and has exactly two critical points. Let $[a_n, b_n]$ be the range of $\lambda_n$ on $\mathcal{O}_\alpha$. Then for each $n$,
\[a_n< b_n<a_{n+1}<b_{n+1}.\]
These $a_n, b_n$, $n=0,1,2,\cdots$ are exactly roots of the following equation:
\begin{equation}
(y_1-\dot{y}_2)^2+4=(\dot{y}_2+y_1+2y_2\tan\frac{\alpha}{2})^2.
\end{equation}
\end{theorem}
\begin{proof}Let $t=\cos \frac{\alpha}{2}\sin \tau$, $\tau\in[-\frac{\pi}{2},\frac{\pi}{2}]$. Then the characteristic equation becomes
\[2\cos(\gamma+\frac{\pi}{2})\cos \tau+(1-\sin \tau)\dot{y}_2+(1+\sin \tau)y_1+2y_2\tan\frac{\alpha}{2}=0.\]
Then the argument in the proof of Thm.~\ref{M1} still holds. We omit the details here.
\end{proof}

For the orbit $\mathcal{O}_\alpha$ with $\alpha\in(\pi, 2\pi)$, the corresponding characteristic equation is
 \begin{equation}(\cos \frac{\alpha}{2}+t)\dot{y}_2+(\cos\frac{\alpha}{2}-t)y_1+2\sin \frac{\alpha}{2} y_2=2\cos (\gamma+\frac{\pi}{2})\sqrt{\cos^2\frac{\alpha}{2}-t^2}.\label{ch}\end{equation}
 \begin{lemma}\label{ii}On the orbit $\mathcal{O}_\alpha$ with $\alpha\in(\pi, 2\pi)$, $\lambda_0$ is bounded from below. In particular, by the continuity principle, $\lambda_n$ is continuous on $\mathcal{O}_\alpha$.
\end{lemma}
\begin{proof}We only need to prove that for sufficiently negative $\lambda$, Eq.~(\ref{ch}) cannot hold for any $\gamma$ and $t$. For this purpose, we should use the following estimations for sufficiently negative $\lambda=-s^2(s>0)$:
\[
y_1(1, \lambda)=\cosh s+O(\frac{e^s}{s}),\quad y_2(1, \lambda)=\frac{\sinh s}{s}+O(\frac{e^s}{s^2}), \quad \dot{y}_2(1, \lambda)=\cosh s+O(\frac{e^s}{s}).
\]
These results are not hard to obtain from Lemma~2.1.1 and Lemma~2.1.2 of \cite[Chap.~2]{Kong3}. Note that unlike in the previous situation, the continuity of $q$ is used to obtain these estimations.

Divide the LHS of Eq.~(\ref{ch}) by $\cosh s$. Then for sufficiently large $s$, the result $< \cos \frac{\alpha}{2}$. Divide the RHS of Eq.~(\ref{ch}) by $\cosh s$. Then for sufficiently large $s$, the result $> \cos \frac{\alpha}{2}$. This is exactly what we want.

\end{proof}
\noindent\emph{Remark.} For $t<-\cos\frac{\alpha}{2}$, the associated quadratic form is
\[Q_1(y)=\int_0^1|y'|^2 dx+\int_0^1q(x)|y|^2dx+\frac{2\sin \frac{\alpha}{2}}{\cos \frac{\alpha}{2}+t}|y(0)|^2, y\in H_{\gamma,t}^1, \]
where \[H_{\gamma,t}^1=\{y\in H^1|y(1)=e^{-i(\gamma+\frac{\pi}{2})}\sqrt{\frac{-\cos \frac{\alpha}{2}+t}{-\cos \frac{\alpha}{2}-t}}y(0)\}\subset H^1.\]
The argument in the proof of Lemma \ref{i} fails to hold. The situation is similar for $t=-\cos\frac{\alpha}{2}$. This is the reason why we have turned to the several estimations in the proof of the above lemma.

\begin{theorem}\label{M3}Assume that $\mathcal{O}_\alpha$ with $\alpha\in (\pi, 2\pi)$ has no joint point with $\Gamma$. Then on $\mathcal{O}_\alpha$, $\lambda_n$ is real analytic, and has exactly two critical points. Let $[a_n, b_n]$ be the range of $\lambda_n$ on $\mathcal{O}_\alpha$. Then for each $n$,
\[a_n< b_n<a_{n+1}<b_{n+1}.\]
These $a_n, b_n$, $n=0,1,2,\cdots$ are exactly roots of the following equation:
\begin{equation}
(y_1-\dot{y}_2)^2+4=(\dot{y}_2+y_1+2y_2\tan\frac{\alpha}{2})^2.
\end{equation}
\end{theorem}
\begin{proof}
The proof is similar to that of Thm.~\ref{M2} and we omit the details.
\end{proof}
\noindent\emph{Remark.} In \cite{Eas}, the authors obtained a general inequality among eigenvalues of different coupled boundary conditions. In fact, these boundary conditions lie on the circle parameterized by $\gamma$ in our adjoint orbit $\mathcal{O}$. In \cite{Bin}, this inequality was re-derived via variational characterization of eigenvalues. To certain extent, our inequality $a_n< b_n<a_{n+1}<b_{n+1}$ can be viewed as an extension in this direction--we consider an adjoint orbit rather than a circle in it.
\begin{example}
Let $q\equiv 0$. Then for $\lambda>0$, the equation in Thm.~\ref{M2} or Thm.~\ref{M3} is
\[\cos \sqrt{\lambda}+\frac{\sin \sqrt{\lambda} }{\sqrt{\lambda}}\tan \frac{\alpha}{2}=\pm1.\]
The critical points are $t=0$ and $\gamma=0$ or $\pi$.
\end{example}
As a conclusion of this section, we shall find out the level set $\Lambda^\kappa$ in an adjoint orbit $\mathcal{O}$ consisting of boundary conditions with $\kappa$ as an eigenvalue.
\begin{theorem}\label{S}Let $\mathcal{O}$ be an adjoint orbit and $p\in \mathcal{O}$. If $\lambda_n(p)=\kappa$ for some $n$, then the level set $\Lambda^\kappa$ is a set either consisting of a single point or diffeomorphic to a circle.
\end{theorem}
\begin{proof}
Let $\zeta_1\neq\zeta_2$ be the two eigenvalues of $\mathcal{O}\subset \mathrm{U}(2)$ and $\varrho_1\neq\varrho_2$ the two eigenvalues of $\Gamma(\kappa)$. The general element of $\mathcal{O}$ is of the following form:
\[U(x, \gamma)=\left( \begin{array}{cc}
\zeta_1 x+\zeta_2(1-x) & (\zeta_2-\zeta_1)\sqrt{x(1-x)}e^{i\gamma}\\
 (\zeta_2-\zeta_1)\sqrt{x(1-x)}e^{-i\gamma}& \zeta_1(1-x)+\zeta_2x \\
 \end{array}
    \right),\]
    where $x\in[0,1]$, $\gamma\in[0,2\pi)$. $\Lambda^\kappa\subset\mathcal{O}$ is characterized by the equation $\det(U(x,\gamma)-\Gamma(\kappa))=0$. Since we are only interested in the shape of $\Lambda^\kappa$, we can safely set $\Gamma(\kappa)=\left( \begin{array}{cc}
\varrho_1 & 0\\
 0& \varrho_2 \\
 \end{array}
    \right)$. This implies the following:
    \[x=-\frac{(\zeta_2-\varrho_1)(\zeta_1-\varrho_2)}{(\varrho_1-\varrho_2)(\zeta_2-\zeta_1)}.\]
    \[1-x=\frac{(\zeta_1-\varrho_1)(\zeta_2-\varrho_2)}{(\varrho_1-\varrho_2)(\zeta_2-\zeta_1)},\]

    If at least one of $\zeta_1, \zeta_2$ coincides with $\varrho_1$ or $\varrho_2$, then $x$ or $1-x$ equals zero and $\Lambda^\kappa$ consists of the only point $p=\left( \begin{array}{cc}
\zeta_1 & 0\\
 0& \zeta_2 \\
 \end{array}
    \right)$ or $\left( \begin{array}{cc}
\zeta_2 & 0\\
 0& \zeta_1 \\
 \end{array}
    \right)$.

    If $\zeta_1,\zeta_2$ are both different from $\varrho_1$ and $\varrho_2$, then both $x$ and $1-x$ are nonzero and determined by these values. It follows that $\Lambda^\kappa$ is diffeomorphic to $S^1$, parameterized by $\gamma$.
\end{proof}
\noindent\emph{Remark}. If $\mathcal{O}$ has no joint point with $\Gamma$, then from our previous result, $\Lambda^\kappa$ is actually the level-$\kappa$ set of $\lambda_n$. If $\kappa=a_n$ or $b_n$, then $\Lambda^\kappa$ consists of the minimizer or maximizer of $\lambda_n$. For other values of $\kappa$, $\Lambda^\kappa$ are all diffeomorphic to $S^1$.

\section{$\lambda_n$ as functions on the boundary circle of $\mathrm{H}/W$}
In the previous sections, we have mainly analyzed the behavior of $\lambda_n$ as functions on generic adjoint orbits represented by interior points in $\mathrm{H}/W$. However, attention should also be paid to points on the boundary circle $\partial(\mathrm{H}/W)$--important boundary conditions, such as Dirichlet, Neumann and Robin boundary conditions, lie on this circle. In this section, we shall consider $\lambda_n$ as functions on $\partial(\mathrm{H}/W)$. Note that $\partial(\mathrm{H}/W)$ can be naturally viewed as the diagonal circle $S_D$ of $\mathrm{H}$, consisting of matrices of the form $e^{i\theta}I$.

It is known that the range of $\lambda_n$ on $\mathrm{U}(2)$ is the same as that of $\lambda_n$ on $\mathrm{H}$ \cite{Kong1}, and the range is closely related to the eigenvalues of the Dirichlet boundary condition. Since $\mathrm{H}$ is 2-dimensional, it's possible to determine this range by restricting $\lambda_n$ on $S_D$. In fact, we have
\begin{theorem}\label{d}The range of $\lambda_n$ on $\mathrm{U}(2)$ is the same as that of $\lambda_n$ on $S_D$. More precisely, if $\Lambda_{n,\kappa}$ is the n-th level-$\kappa$ curve on $\mathrm{H}$, then $S_D$ intersects $\Lambda_{n,\kappa}$ at a unique point.
\end{theorem}
\begin{proof}
The proof is based on Thm.~2.2 in \cite{Kong3}, where in fact the level curve $\Lambda_{n, \kappa}\subset \mathrm{H}$ is characterized.

In \cite{Kong3}, boundary conditions in $\mathrm{H}$ are written in the form of Eq.~(\ref{B1}). It is easy to find that the diagonal of $\mathrm{H}$ corresponds to $\alpha, \beta$'s satisfying $\alpha+\beta=\pi$. The level curve $\Lambda_{n, \kappa}\subset \mathrm{H}$ can be written as
\[\{(\alpha,\beta)\in [0, \pi)\times (0,\pi]|\alpha=f(\beta), \beta\in J_0\},\]
where the precise form of the interval $J_0\subset (0, \pi]$ depends on whether $\kappa>\lambda_{n-1}^D$ or not, and the function $f$ is strictly increasing on $J_0$. So if $S_D$ intersects $\Lambda_{n, \kappa}$, the intersection point is unique. As for the existence of the intersection point, it can be derived easily from the argument of \cite{Kong3}, cf. Fig.1 there.
\end{proof}
\noindent \emph{Remark}. It should be pointed out that in \cite{Kong3} there is another "diagonal" $\mathscr{C}$ in $\mathrm{H}$ (see Eq.~(1.10) in \cite{Kong3}), which is different from ours. $\mathscr{C}$ corresponds to $\alpha, \beta$'s satisfying $\alpha=\beta$, rather than $\alpha+\beta=\pi$. Thm.~\ref{d} does not hold when $S_D$ is replaced by $\mathscr{C}$.
\begin{corollary}If $S_D$ is parameterized by $\beta\in(0, \pi]$ (so $\alpha=\pi-\beta$), then for each $n$, $\lambda_n$ as a function of $\beta$ is strictly increasing and continuous.
\end{corollary}
\begin{proof}
That $\lambda_n$ is continuous can be derived from Lemma 2.1 in \cite{Bin}. For $\beta\in (0, \pi)$, the associated quadratic form is
\[Q_\beta(y)=\int_0^1|y'|^2 dx+\int_0^1q(x)|y|^2dx-\cot \beta |\psi|^2,\quad y\in H^1\]
and the strict monotonicity of $\lambda_n$ is a conclusion of the min-max principle and Thm.~\ref{d}.
\end{proof}
\section*{Acknowledgemencts}
 This study is supported by the Fundamental Research Funds for the Central Universities (2014B14514).

\end{document}